\documentclass[journal,onecolumn,12pt]{IEEEtran}
\usepackage{pslatex}
\usepackage{amsfonts,color,morefloats}
\usepackage{amssymb,amsmath,latexsym,amsthm,mathrsfs}
\usepackage{cases}

\newcommand{\Z}{\mathbb{{Z}}}
\newcommand{\Q}{\mathbb{{Q}}}

\newcommand{\tr}{{\mathrm{Tr}}}

\newcommand{\gf}{{\mathrm{GF}}}
\newcommand{\C}{{\mathcal{C}}}

\renewcommand{\c}{\mathbf{c}}
\renewcommand{\t}{\mathbf{t}}
\renewcommand{\v}{\mathbf{v}}
\newcommand{\Comp}{\mathrm{Comp}}
\newcommand{\wt}{\mathrm{wt}}
\newcommand{\At}{A_\mathbf{t}}
\newcommand{\Bt}{B_\mathbf{t}}
\newcommand{\Uq}{U_{q+1}}

\newcommand{\z}{\mathbf{z}}

\newcommand{\bP}{{\mathbb{P}}}

\newcommand{\bc}{{\mathbf{c}}}
\newcommand{\fA}{{\mathscr{A}}}
\newcommand{\fB}{{\mathscr{B}}}

\newcommand{\bx}{{\mathbf{x}}}
\newcommand{\by}{{\mathbf{y}}}
\newcommand{\bt}{{\mathbf{t}}}
\newcommand{\bz}{{\mathbf{z}}}
\newcommand{\bone}{{\mathbf{1}}}

\newcommand{\bu}{{\mathbf{u}}}

\newtheorem{theorem}{Theorem}
\newtheorem{lemma}[theorem]{Lemma}

\newtheorem{example}{Example}

\newtheorem{rmk}{Remark}

\begin{document}
\title{On correlation distribution of Niho-type decimation $d=3(p^m-1)+1$}
\date{\today}
\author{Maosheng Xiong, Haode Yan\thanks{M. Xiong was supported by the Research Grants Council (RGC) of Hong Kong (No. 16306520). H. Yan's research was supported by the Natural Science Foundation of Sichuan Province (Grant No. 2022NSFSC1805) and the Fundamental Research Funds for the Central Universities of China (Grant No. 2682023ZTPY002). (\emph{Corresponding author: Haode Yan}).

M. Xiong is with the Department of Mathematics, The Hong Kong University of Science and Technology, Hong Kong (e-mail: mamsxiong@ust.hk).

H. Yan is with the School of Mathematics, Southwest Jiaotong University, Chengdu 610031, China (e-mail: hdyan@swjtu.edu.cn).}
}
\maketitle
\begin{abstract}

The cross-correlation problem is a classic problem in sequence design. In this paper we compute the cross-correlation distribution of the Niho-type decimation $d=3(p^m-1)+1$ over $\gf(p^{2m})$ for any prime $p \ge 5$. Previously this problem was solved by Xia et al. only for $p=2$ and $p=3$ in a series of papers. The main difficulty of this problem for $p \ge 5$, as pointed out by Xia et al., is to count the number of codewords of ``pure weight'' 5 in $p$-ary Zetterberg codes. It turns out this counting problem can be transformed by the MacWilliams identity into counting codewords of weight at most 5 in $p$-ary Melas codes, the most difficult of which is related to a K3 surface well studied in the literature and can be computed. When $p \ge 7$, the theory of elliptic curves over finite fields also plays an important role in the resolution of this problem.
\end{abstract}
	

\section{Introduction}
Let $p$ be a prime, $n$ a positive integer and $\{s(t)\}$ be an $m$-sequence of period $p^n-1$ over the finite field $\gf(p)$ of order $p$. Let $d$ be an integer. The $d$-decimation sequence of $\{s(t)\}$ is defined as the sequence $\{s(dt)\}$. Here $dt$ is taken modulo $p^n-1$ for $t=0,1, \cdots, p^n-2$, and $d$ is called a decimation. If $(d,p^n-1)=1$, then the $d$-decimation sequence $\{s(dt)\}$ is also an $m$-sequence of period $p^n-1$. The cross correlation function $C_d(\tau)$ between the sequence $\{s(t)\}$ and its $d$-decimation sequence $\{s(dt)\}$ is given by
\begin{eqnarray*} \label{1:cdt} C_d(\tau)=\sum_{t=0}^{p^n-2}\omega_p^{s(t+\tau)-s(dt)},\end{eqnarray*}
where $\tau=0,1,\cdots,p^n-2$ and $\omega_p=\exp(2 \pi \sqrt{-1}/p)$ is a primitive complex $p$-th root of unity. In the theory of sequence design, it is interesting to 1). find new decimation $d$ that leads to low cross-correlation and 2). determine the values $C_d(\tau)$ together with the frequency of the occurrence of each value. This is a classical problem in sequence design and is known as the correlation distribution for the decimation $d$. Because of important applications, this problem has received a lot of attention since the 1960s, and many interesting theoretical results have been obtained \cite{G,Gong,Tor76,Tor78,Tor98,Niho,R,T,ZLFG}.

In 1972 Niho studied the cross correlation function between two $m$-sequences of period $2^{2m}-1$ differing by a decimation of the form $d=s(2^m-1)+1$ for some  positive integers $s$ and $m$ in his influential Ph.D thesis \cite{Niho} in which Niho computed the cross-correlation distribution for many such decimations and proposed many conjectures and open questions. Such problems have become a source of inspiration for researchers for the next 50 years, some of which remain open even today. Due to his remarkable work and also in honor of him, the decimation studied in the thesis is called Niho type decimation. Niho's work has been extended to odd characteristic \cite{R}, and consequently, let $p$ be any prime, a positive integer $d$ is call a Niho exponent with respect to the finite field $\gf(p^{2m})$ if $d$ is of the form
\begin{eqnarray} \label{1:d} d=s(p^m-1)+1, \quad \gcd(d,p^{2m}-1)=1\end{eqnarray}
for some integer $s$. The GCD condition is to ensure that the $d$-decimation sequence is also an $m$-sequence. Since 1972, Niho type exponents have found important applications in many areas such as sequence design, cryptography and coding theory. Interested readers may refer to the survey paper \cite{LZ} on applications and open problems related to Niho exponents.

In this paper we focus on the Niho exponent $d$ given in (\ref{1:d}) with $s=3$. When $p=2$, it was known to Niho that the cross correlation function $C_d(\tau)$ takes at most $6$ distinct values (see \cite[Theorem 3-9]{Niho}), and the value distribution of $C_d(\tau)$ was computed thirty years later by Dobbertine et al. \cite{DobF} who employed sophisticated techniques involving Dickson polynomials and Kloosterman sums. The final result of \cite{DobF} is, however, not quite satisfactory in the sense that the expressions obtained depend on an unknown term, which involves the Kloosterm sum in a complicated way and is difficult to evaluate. Then finally in 2016, Xia et al. \cite{XiaLi} resolved this issue by finding a close connection of this problem to binary Zetterberg codes from which the value distribution of $C_d(\tau)$ can be explicitly determined. When $p$ is odd, there is a similar story: it was known that $C_d(\tau)$ takes at most $6$ distinct values, and Xia et al. \cite{XiaLi2} also found that the value distribution $C_d(\tau)$ is closely related to $p$-ary Zetterberg codes. By studying $3$-ary Zetterberg codes explicitly and directly, they solved the value distribution problem for the case $p=3$. However, for cases $p \ge 5$, the problem seems much more difficult, and this was proposed as an open problem in \cite{LZ}.

{\bf Open problem.} Determine the value distribution of $C_d(\tau)$ for $d=3(p^m-1)+1$ where $p\geq 5$.

In this paper we solve this problem. The main result of this paper is presented as follows.
\begin{theorem} \label{1:main} Let $p \ge 5$ be a prime, $d=3(p^m-1)+1$ and $\gcd(d,p^{2m}-1)=1$, the value distribution of $C_d(\tau)$ is given by
	\begin{eqnarray*}
		C_d(\tau)=
		\left\{ \begin{array}{llll}
			-p^m-1&\mbox{occurs}&\frac{1}{2}p^{2m}-\frac{1}{3}p^m-\frac{1}{6}b_3p^m+N_4+4N_5 &\mbox{times} \\
			-1&\mbox{occurs}&-\frac{1}{2}p^m+\frac{1}{2}b_3p^m-4N_4-15N_5-1&\mbox{times}\\
			p^m-1&\mbox{occurs}&\frac{1}{2}p^{2m}+p^m-\frac{1}{2}b_3p^m+6N_4+20N_5&\mbox{times}\\
			2p^m-1&\mbox{occurs}&-\frac{1}{6}p^m+\frac{1}{6}b_3p^m-4N_4-10N_5&\mbox{times}\\
			3p^m-1&\mbox{occurs}&N_4&\mbox{times}\\
			4p^m-1&\mbox{occurs}&N_5&\mbox{times}.
		\end{array}\right.
	\end{eqnarray*}
	Here \begin{eqnarray*}
		b_3=\left\{\begin{array}{lll}
			p^m, ~~~~~~\mathrm{if}~~ p^m \not \equiv 2 \pmod{3},\\
			p^m+2, ~~\mathrm{if}~~ p^m \equiv 2 \pmod{3},
		\end{array} \right. \end{eqnarray*}
	and if $p=5$, then
	\begin{eqnarray*}
		N_4&=&\big(p^m-(-1)^m\big)/6,\\
		N_5&=&\left(p^{2m}-p^m\left(7+4(-1)^m\right)+10(-1)^m\right)/120,
	\end{eqnarray*}
	and if $p \ge 7$, then
	\begin{eqnarray*}
		N_4&=&\left(p^m-4+3\left(\frac{5}{p}\right)^m+3\left(\frac{-15}{p}\right)^m+\lambda_{p^m}\right)/6,\\
		N_5&=&\left(p^{2m}-p^m \left(6+4\left(\frac{p}{3}\right)^m \right)+16-20\left(\frac{5}{p}\right)^m-15\left(\frac{-15}{p}\right)^m-10 \lambda_{p^m}-A_{p^m}\right)/120.
	\end{eqnarray*}
	Here $\left(\frac{\cdot}{\cdot}\right)$ is the Jacobi symbol, $\lambda_{p^m}$ is given by Theorem \ref{pre:thm1} in Section \ref{pre}, and $A_{p^m}$ is given by Theorem \ref{5:nqxb} in Section \ref{appen} {\bf Appendix}.
\end{theorem}

\begin{rmk} When $p \ge 7$, the term $\lambda_{p^m}$ is a character sum related to an elliptic curve over $\gf(p)$ (see Theorem \ref{pre:thm1}), and the term $A_{p^m}$ is related to a K3 surface (see Theorem \ref{5:nqxb}). \end{rmk}

By using {\bf Magma}, we compute some examples for small $p$ and $m$. These results are consistent with Theorem \ref{1:main}.
\begin{example}When $p=5$ and $m=2$, then $\gcd(d,5^4-1)=1$, we have $b_3=25$. Then the value distribution of $C_d(\tau)$ is given by
	\begin{eqnarray*}
		C_d(\tau)=
		\left\{ \begin{array}{llll}
			-26&\mbox{occurs}&216 &\mbox{times} \\
			-1&\mbox{occurs}&238&\mbox{times}\\
			24&\mbox{occurs}&109&\mbox{times}\\
			49&\mbox{occurs}&54&\mbox{times}\\
			74&\mbox{occurs}&4&\mbox{times}\\
			99&\mbox{occurs}&3&\mbox{times}.
		\end{array}\right.
	\end{eqnarray*}
\end{example}	
\begin{example}
	When $p=7$ and $m=3$, then $\gcd(d,7^6-1)=1$, we have $b_3=343$, $\left(\frac{p}{3}\right)=1$, $\left(\frac{5}{p}\right)=-1$, $\left(\frac{-15}{p}\right)=-1$, $\lambda_{p^m}=-21$  and $A_{p^m}=0$. Then the value distribution of $C_d(\tau)$ is given by
	\begin{eqnarray*}
		C_d(\tau)=
		\left\{ \begin{array}{llll}
			-344&\mbox{occurs}&42970&\mbox{times} \\
			-1&\mbox{occurs}&44134&\mbox{times}\\
			342&\mbox{occurs}&19735&\mbox{times}\\
			685&\mbox{occurs}&9803&\mbox{times}\\
			1028&\mbox{occurs}&52&\mbox{times}\\
			1371&\mbox{occurs}&954&\mbox{times}.
		\end{array}\right.
	\end{eqnarray*}
\end{example}	
\begin{example}When $p=11$ and $m=1$, then $\gcd(d,11^2-1)=1$, we have $b_3=13$, $\left(\frac{p}{3}\right)=-1$, $\left(\frac{5}{p}\right)=1$, $\left(\frac{-15}{p}\right)=-1$, $\lambda_{p^m}=-1$  and $A_{p^m}=0$. Then the value distribution of $C_d(\tau)$ is given by
	\begin{eqnarray*}
		C_d(\tau)=
		\left\{ \begin{array}{llll}
			-12&\mbox{occurs}&38 &\mbox{times} \\
			-1&\mbox{occurs}&46&\mbox{times}\\
			10&\mbox{occurs}&26&\mbox{times}\\
			21&\mbox{occurs}&8&\mbox{times}\\
			32&\mbox{occurs}&1&\mbox{times}\\
			43&\mbox{occurs}&1&\mbox{times}.
		\end{array}\right.
	\end{eqnarray*}
\end{example}	
Here we give an overview of the main ideas in proving Theorem \ref{1:main}. Since $C_d(\tau)$ takes at most $6$ distinct values, it suffices to determine $6$ unknown quantities $N_0,N_1,\cdots,N_5$ which will be defined later. By using the moment methods, four identities involving these $N_i$'s can be obtained easily (see Lemma \ref{lem1:p}), so we just need to compute two more. It turns out that $N_4$ and $N_5$ are both directly related to counting codewords of certain weight in a $p$-ary Zetterberg code. Similar to \cite{DobF,XiaLi,XiaLi2} for $p=2,3$, we can compute $N_4$ directly for $p \ge 5$, invoking the theory of elliptic curves to evaluate a certain character sum along the way (see Theorem \ref{N4:con}). The most difficult step of the paper is to compute $N_5$, which is to count the number of ``pure weight'' 5 codewords in a $p$-ary Zetterberg code. This is quite non-trivial even for $p=3$ (see \cite[Lemma 9]{XiaLi2}). Our treatment of $N_5$ also represents the most innovation of the paper: by using complete weight distribution of $p$-ary Zetterberg codes and the MacWilliams identity, this counting problem can be transformed into counting codewords of small weight of certain patterns in $p$-ary Melas codes, most of which can be computed directly and explicitly except for one term. This one term, while difficult on its own, turns out to be directly related to a K3 surface well studied in the literature \cite{Peter} and hence can be computed (see Theorem \ref{N5:con}). With the values $N_0,N_1,\cdots, N_5$ in hands, then we can obtain Theorem \ref{1:main} easily.

This paper is organized as follows. In Section \ref{pre} we provide some preliminary results regarding character sums, elliptic curves, the complete weight distribution and MacWilliams identity that will be used later. In Section \ref{sec:dis} we start proving Theorem \ref{1:main} by first establishing four moment identities involving quantities $N_0,N_1,\cdots,N_5$ (see Lemma \ref{lem1:p}). In Section \ref{sec:N4} we compute $N_4$ directly (see Theorem \ref{N4:con}). The computation of $N_5$ is the most difficult and is divided in two Sections: in Section \ref{MZcodes} we introduce Melas codes, Zetterberg codes and the ``pure weight'' distribution which are all necessary to treat $N_5$. Armed with this knowledge, then in Section \ref{sec:N5} we compute $N_5$ (see Theorem \ref{N5:con}), and hence we prove Theorem \ref{1:main}. In computing $N_5$, we need quite a few counting results regarding the number of small weight codewords of various patterns in $p$-ary Melas codes. The proof of these results is a little technical. To streamline the ideas of the paper, we put these technical results in Section \ref{appen} {\bf Appendix}. In Section \ref{sec:cond} we conclude this paper.

\section{Preliminaries}\label{pre}
	Throughout the paper, let $p$ be a prime and $q=p^m$ for some positive integer $m$. For convenience, we also adopt the convention that $\frac{1}{0}:=0$.

\subsection{Character sums and elliptic curves}

Assume here that $p$ is odd. Denote by $\eta$ the quadratic character on $\gf(p^m)^*:=\gf(p^m) \setminus \{0\}$.  It is convenient to extend the definition of $\eta$ to $\gf(p^m)$ by setting $\eta(0)=0$. For any element $\beta\in \gf(p^m)$ such that $\eta(\beta)=1$, there are two distinct square roots of $\beta$ in $\gf(p^m)$, which are denoted by $\pm\sqrt{\beta}$.

For a polynomial $f(x)\in\gf(p^m)[x]$, we shall consider character sums of the form \begin{equation*} \label{2:cha-f} \sum_{x\in\gf(p^m)}\eta(f(x)).
\end{equation*}
This character sum is trivial to evaluate when $f(x)$ is linear. When $f(x)$ is quadratic, we have
\begin{lemma} \cite[Theorem 5.48]{FF}\label{charactersumquadratic} Let $f(x)=a_2x^2+a_1x+a_0\in\gf(p^m)[x]$ with $p$ odd and $a_2\neq0$. Put $d=a^2_1-4a_0a_2$ and let $\eta$ be the quadratic character of $\gf(p^m)$. Then
	\begin{equation*}
		\sum_{x\in\gf(p^m)}\eta(f(x))=
		\begin{cases}
			-\eta(a_2), & \text{ if } d\neq 0, \\
			(p^m-1)\eta(a_2), & \text{ if } d = 0.
		\end{cases}
	\end{equation*}
\end{lemma}

As it will be seen later, to compute the correlation distribution for $d=3(p^m-1)+1$, we need to evaluate a specific character sum
\begin{eqnarray}
	\label{2:f1} \lambda_{p^m}:=\sum_{x\in\gf(p^m)}\eta\left((x^2-4)(2x+1)(2x+5)\right).\end{eqnarray}
The value $\lambda_{5^m}$ can be computed easily. When $p \ge 7$, the polynomial involved in $\lambda_{p^m}$ is of degree 4, so the character sum corresponds to the elliptic curve $y^2=(x^2-4)(2x+1)(2x+5)$ over $\gf(p^m)$. The evaluation of $\lambda_{p^m}$ is equivalent to the computation of the Hasse zeta function of the corresponding elliptic curve over $\gf(p)$. Generally speaking, there is no explicit and simple formula for such character sums, except in rare situation when the elliptic curves are very special \cite{SilveEC}. Based on the theory of elliptic curves, we can employ an efficient algorithm to evaluate $\lambda_{p^m}$ for any $p \ge 7$. Interested readers may also review \cite[Section II]{Yan} for a more detailed description as the theory of elliptic curves was used in the same way to evaluate a similar character sum for $p \ge 5$.

\begin{theorem}\label{pre:thm1}
	\begin{itemize}
		\item[(i)] If $p=5$, then
		\begin{eqnarray} \label{pre:lam} \lambda_{5^m}=-1-(-1)^m.
		\end{eqnarray}
		
		\item[(ii)] If $p\geq 7$. Denote by $N_p$ the number of $(x,y) \in \gf(p)^2$ satisfying the equation
		\begin{equation}\label{ec}
			E: y^2=x(x-5)(x+27).
		\end{equation}
		Define $a=N_{p}-p$ and let $\alpha$ and $\beta$ be the two complex roots of the quadratic polynomial $T^2+aT+p$. Then
		\begin{eqnarray*} \label{2:eva-cha}
			\lambda_{p^m}&=&-1-\alpha^m-\beta^m.
		\end{eqnarray*}
	\end{itemize}
\end{theorem}

\begin{proof}
	(i) When $p=5$, it is easy to see that
	\begin{eqnarray*}
		\lambda_{5^m}&=&\sum_{x\in\gf(5^m)}\eta\left((x-2)^2x(x+2)\right)
		=\sum_{x\in\gf(5^m) \setminus \{2\}}\eta\left(x(x+2)\right)\\
		&=&\sum_{x\in\gf(5^m)}\eta\left(x(x+2)\right)-\eta(8).
	\end{eqnarray*}
	The first term is $-1$ by Lemma \ref{charactersumquadratic}, and the second term is $(-1)^m$. This proves (\ref{pre:lam}).
	
	(ii) When $p \ge 7$, Equation (\ref{ec}) defines an elliptic curve $E$ over $\gf(p)$. The character sum  $\Gamma_{p,\,m}$ given by
	\[\Gamma_{p,\,m}=\sum_{x \in \gf(p^m)} \eta\left(x(x-5)(x+27)\right)\]
	is closely related to the number of $\gf(p^m)$-rational points (with an extra point at infinity) on $E$, which is actually $p^m+1+\Gamma_{p,\,m}$. By the theory of elliptic curves (see \cite[Theorem 2.3.1, Chap. V]{SilveEC}), we have
	\[\Gamma_{p,\,m}=-\alpha^m-\beta^m.\]
	Now we can use $\Gamma_{p,\,m}$ to evaluate $\lambda_{p^m}$. Replacing $2x$ by $x$ in $\lambda_{p^m}$, we have
	\begin{eqnarray*}
		\lambda_{p^m}&=&\sum_{x\in\gf(p^m)}\eta\left((x^2-16)(x+1)(x+5)\right)\\
		&=&\sum_{x\in\gf(p^m)\setminus\{-1\}}\eta\left((x-4)(x+4)(x+1)(x+5)\right).
	\end{eqnarray*}
	The relation $x=-\frac{y+1}{y}$ gives a $1$-to-$1$ correspondence between $x\in\gf(p^m)\setminus\{-1\}$ and $y\in\gf(p^m)\setminus\{0\}$, thus
	\begin{eqnarray*}
		\lambda_{p^m} &=&\sum_{y\in\gf(p^m)\setminus\{0\}}\eta\left(\left(-\frac{y+1}{y}-4\right)
		\left(-\frac{y+1}{y}+4\right)\left(-\frac{y+1}{y}+1\right)\left(-\frac{y+1}{y}+5\right)\right)\\
		&=&\sum_{y\in\gf(p^m)\setminus\{0\}}\eta\left(\frac{(3y-1)(4y-1)(5y+1)}{y^4}\right)\\
		&=&-1+\sum_{y\in\gf(p^m)}\eta\left((3y-1)(4y-1)(5y+1)\right).
	\end{eqnarray*}
	Since $p \ge 7$, letting $d=3 \cdot 4 \cdot 5$ and replacing $d y$ by $y$ we have
	\begin{eqnarray*}
		\lambda_{p^m}&=&-1+\sum_{y\in\gf(p^m)}\eta\left((d y-20)(d y-15)(d y+12)\right)\\
		&=&-1+\sum_{y\in\gf(p^m)}\eta\left((y-20)(y-15)(y+12)\right).
	\end{eqnarray*}
	Now making a change of variable $y=z+15$, we obtain the desired result
	\begin{eqnarray*}
		\lambda_{p^m}&=&-1+\sum_{z\in\gf(p^m)}\eta\left(z(z-5)(z+27)\right)=-1+\Gamma_{p,\,m}.
	\end{eqnarray*}
\end{proof}
Here as examples we compute the value $\lambda_{p^m}$ for a few prime $p$. 

\begin{example}
	When $p=7$, it is easy to check that $N_7=11$ where $N_p$ is the number $(x,y) \in \gf(p)^2$ satisfying the equation $E:y^2=x(x-5)(x+27)$, so $a=N_7-7=4$, the roots of $T^2+4T+7$ are
	\[T_{1,2}=-2 \pm \sqrt{-3},\]
	so
	\[\lambda_{p^m}=-\left(-2 + \sqrt{-3}\right)^m-\left(-2 - \sqrt{-3}\right)^m-1. \]
\end{example}

\begin{example}
	When $p=11$, it is easy to check that $N_{11}=11$, so $a=N_{11}-11=0$, the roots of $T^2+11$ are
	\[T_{1,2}=\pm \sqrt{-11},\]
	so
	\[\lambda_{p^m}=-\left(\sqrt{-11}\right)^m-\left( - \sqrt{-11}\right)^m-1. \]
\end{example}

\begin{example}
	When $p=13$, it is easy to check that $N_{13}=11$, so $a=N_{13}-13=-2$, the roots of $T^2-2T+13$ are
	\[T_{1,2}=1 \pm 2\sqrt{-3},\]
	so
	\[\lambda_{p^m}=-\left(1 + 2\sqrt{-3}\right)^m-\left(1 - 2\sqrt{-3}\right)^m-1. \]
\end{example}

\begin{example}
	When $p=17$, it is easy to check that $N_{17}=11$, so $a=N_{17}-17=-6$, the roots of $T^2-6T+17$ are
	\[T_{1,2}=3 \pm 2\sqrt{-2},\]
	so
	\[\lambda_{p^m}=-\left(3 + 2\sqrt{-2}\right)^m-\left(3 - 2\sqrt{-2}\right)^m-1. \]
\end{example}

\subsection{Linear codes, the complete weight distribution and MacWilliams identity}

Let $\C$ be an $[n,k]_q$ code, that is, $\C$ is a $k$-dimensional subspace of $\gf(q)^n$ over $\gf(q)$. Let $\C^\bot$ be the dual code of $\C$ with respect to the usual inner product on $\gf(q)^n$.

To define the complete weight distribution of $\C$, we first order elements of $\gf(q)$ explicitly as $u_0=0,u_1,\cdots,u_{q-1}$. Next, for any $\c=(c_1,\ldots,c_n) \in \gf(q)^n$, the \emph{complete weight} of $\c$ is defined to be
$$\Comp(\c):=(w_0,w_1,\cdots,w_{q-1}),$$
where $w_i:=w_i(\c)=\#\{1 \le j \le n: c_j=u_i\}$. It is clear that $w_i$ is a non-negative integer for each $i$ and $\sum_{i=0}^{q-1} w_i=n$.

Denote by $\Sigma^{n,q}$ the set of all integer vectors $\mathbf{t}=(t_0,t_1,\cdots,t_{q-1})$ such that $t_i \ge 0 \, \forall \, i$ and $\sum_{i=0}^{q-1} t_i=n$. The \emph{complete weight enumerator} of the code $\C$ is defined as the polynomial in variables $z_0,\ldots,z_{q-1}$:
\begin{eqnarray*} \label{1:cweight} W_\mathcal{C}(z_0,z_1,\cdots,z_{q-1}):=\sum_{\c \in \mathcal{C}} z_0^{w_0}z_1^{w_1}\cdots z_{q-1}^{w_{q-1}}=\sum_{\mathbf{t}} \At \, \z^\t.\end{eqnarray*}
Here the subscript $\t=(t_0,t_1,\ldots,t_{q-1})$ runs over all $\t \in \Sigma^{n,q}$, $\z^\t:=z_0^{t_0}z_1^{t_1}\cdots z_{q-1}^{t_{q-1}}$, and
\[\At:=\#\left\{\c \in \C: \Comp(\c)=\mathbf{t}\right\}, \quad \forall \mathbf{t} \in \Sigma^{n,q}.\]

\noindent The sequence $(\At)_{\t}$ is called the \emph{complete weight distribution} of $\C$. Setting $z_0=1, z_1=\cdots=z_{q-1}=z$ in the complete weight enumerator of $\C$, we obtain the \emph{weight enumerator} of $\C$, which can be written as
\[W_C(1,z,\cdots,z)=\sum_{t=0}^n A_t z^t, \]
where
\[A_t:=\#\{\c=(c_1,\ldots,c_n) \in \C: \wt(\c)=t\}.\]
Here $\wt(\c)$ is the Hamming weight of $\c$. So the complete weight enumerator of $\C$ provides much more detailed information than the weight enumerator of $\C$ if $q>2$.

For any $\t \in \Sigma^{n,q}$, let
\[\Bt:=\#\left\{\c \in \C^\bot: \Comp(\c)=\t\right\}. \]
The complete weight enumerator of $\C^\bot$ is defined similarly as
\[W_{\C^\bot}(z_0,z_1,\cdots,z_{q-1})=\sum_{\t} \Bt \, \z^\t.\]
Denote by $\psi: \gf(q) \to \mathbb{C}^\times$ the standard additive character. The fundamental relation between the weight enumerators of $\C$ and $\C^\bot$ is given by MacWilliams identity \cite[Chapter 5]{MacW}:
\begin{equation}\label{2:MacW}
	W_{\C^\bot}(z_0,z_1,\cdots,z_{q-1})=\frac{1}{q^k}W_\C(x_0,x_1,\cdots,x_{q-1}),
\end{equation}
where
\begin{eqnarray} \label{2:xi}
	x_i=\sum_{s=0}^{q-1} \psi(u_iu_s)z_s, \quad \forall i=0,1, \cdots, q-1.
\end{eqnarray}
This means that we can derive the complete weight distribution of a linear code from that of its dual code.

	\section{Correlation distribution of $d=3(q-1)+1$}\label{sec:dis}
	From what follows, let $p \ge 2$ be a prime, $m$ a positive integer, $n=2m$ and $q=p^m$, $d=3 \cdot (q-1)+1$. Assume that $\gcd(d,q^2-1)=1$. This is equivalent to $\gcd(5,q+1)=1$. Such $d$ is called a Niho decimation for $\gf(q^2)$.

Let $\alpha$ be a primitive element of $\gf(q^2)$ and $\tr^n_1: \gf(q^2) \to \gf(p)$ the standard trace function. Without loss of generality, we may assume that the $p$-ary $m$-sequence $\{s(t)\}$ is generated by $s(t)=\tr^n_1(\alpha^t)$ for any integer $t$ (see \cite{Niho}). Then the cross correlation function $C_d(\tau)$ between the sequence $\{s(t)\}$ and its $d$-decimation sequence $\{s(dt)\}$ is given by
\begin{eqnarray*} \label{xy4:cdt} C_d(\tau)=\sum_{t=0}^{p^n-2}\omega_p^{\tr_1^n\left(\alpha^{t+\tau}-\alpha^{dt}\right)}=\sum_{x \in \gf(q^2)^*} \omega_p^{\tr_1^n\left(\alpha^{\tau}x-x^d\right)},\end{eqnarray*}
where $\tau=0,1,\cdots,p^n-2$ and $\gf(q^2)^*=\gf(q^2) \setminus \{0\}$.

For simplicity, let us denote
\begin{eqnarray*} \label{1:sa} s(a)=\sum_{x \in \gf(q^2)^*} \omega_p^{\tr^n_1(ax+x^d)}, \quad \forall x \in \gf(q^2). \end{eqnarray*}
It is easy to see that the value distribution of $C_d(\tau)$ for $\tau=0,1,\cdots,p^n-2$ is identical to that of $s(a)$ for $a \in \gf(q^2)^*$. Our strategy is to compute the value distribution of $s(a)$ as $a$ takes values in $\gf(q^2)$, from which the value distribution of $s(a)$ for $a \in \gf(q^2)^*$ can be derived easily. We also remark that this problem has been settled for $p=2$ in \cite{XiaLi} and $p=3$ in \cite{XiaLi2}. Hence our main focus is to study this problem for $p \ge 5$.

Denote
\[\Uq:=\{x \in \gf(q^2): x^{q+1}=x \bar{x}=1\}. \]
Here for any $x \in \gf(q^2)$, we define $\bar{x}=x^q$. The set $\Uq$ is called the unit circle of $\gf(q^2)$. It is a cyclic subgroup of $\gf(q^2)^*$ of order $q+1$ and plays an important role in the treatment of Niho exponents.

For any $a \in \gf(q^2)$, define
\[F_a(z)=z^5+\bar{a}z^3+az^2+1.\]
It was known that
\begin{eqnarray*} \label{1:sa-1} s(a)=q(N(a)-1)-1,\end{eqnarray*}
where $N(a)$ is the number of $z \in \Uq$ such that $F_a(z)=0$. For a proof, interested readers may consult for example \cite{XiaLi} for $p=2$ and \cite{XiaLi2} for any prime $p$.

It is clear that $N(a) \in \{0,1,2,3,4,5\}$. Denote
\[N_i:=\#\{a \in \gf(q^2): N(a)=i\}. \]
To study the value distribution of $s(a)$, it suffices to compute the values $N_i$ for $0 \le i \le 5$.

\begin{lemma} \label{lem1:p}
	We have the following identities:
	\begin{itemize}
		\item[1)] $\sum_{i=0}^5 N_i=q^2$.
		
		\item[2)] $\sum_{i=0}^5 (i-1)N_i=q$.
		
		\item[3)] $\sum_{i=0}^5 (i-1)^2N_i=q^2$.
		
		\item[4)] $\sum_{i=0}^5 (i-1)^3 N_i=q \cdot b_3$,
		where $b_3$ is the number of $y \in \gf(q^2)$ of the equation
		\begin{eqnarray*} \label{lem2:b3}
			(y+1)^d-y^d=1.
		\end{eqnarray*}
	\end{itemize}
\end{lemma}
\begin{proof}
	When $p=2$, 1)-4) were all known (see Lemma 5 in \cite{DobF}). For a general prime $p$, 1) is trivial, and 2)-4) were essentially the same as \cite[Theorem 2.4]{R} where the proofs can be found. We omit the details. \end{proof}

\begin{lemma} \label{lem2:p} The value $b_3$ is given by
	\begin{eqnarray*} \label{lem2:b3}
		b_3=\left\{\begin{array}{lll}
			q&:& q \not \equiv 2 \pmod{3},\\
			q+2&:& q \equiv 2 \pmod{3}.
		\end{array} \right. \end{eqnarray*}
	
\end{lemma}
\begin{proof}
	We need to solve the equation $(y+1)^d-y^d-1=0$ for $y \in \gf(q^2)$, where $d=3q-2$. Obviously $y=0,-1$ satisfy this equation. If $y \in \gf(q^2) \setminus \{0,-1\}$, denoting $\bar{y}:=y^q$ we have
	\[0=(y+1)^d-y^d-1=\frac{(\bar{y}+1)^3}{(y+1)^2}-\frac{\bar{y}^3}{y^2}-1.\]
	We can clear the denominators to have
	\begin{eqnarray*} y^2(\bar{y}+1)^3-(y+1)^2 \bar{y}^3-y^2(y+1)^2=0. \end{eqnarray*}
	Expanding the left hand side of the above equation and collecting terms, we can rewrite it as
	\begin{eqnarray} \label{lem2:b3n}
		(y-\bar{y})^2(y^2+2y\bar{y}+2y+\bar{y})=0.
	\end{eqnarray}
	Because $y=0,-1$ also satisfy the equation (\ref{lem2:b3n}), the quantity $b_3$ is also the number of $y \in \gf(q^2)$ of the equation (\ref{lem2:b3n}). Now we can compute $b_3$ easily.
	
	If $y-\bar{y}=0$, this is equivalent to $y \in \gf(q)$, such $y$'s are solutions of (\ref{lem2:b3n}) and  the total number of such $y$'s is $q$.
	
	Now suppose $y \in \gf(q^2) \setminus \gf(q)$. Then we have
	\begin{eqnarray} \label{lem2:eq2} y^2+2y\bar{y}+2y+\bar{y}=0.
	\end{eqnarray}
	Taking $q$-th power on both sides of (\ref{lem2:eq2}), noting that $\bar{y}^q=y^{q^2}=y$, we have
	\begin{eqnarray} \label{lem2:eq3} \bar{y}^2+2y\bar{y}+2\bar{y}+y=0.
	\end{eqnarray}
	Adding the left hand sides of (\ref{lem2:eq2}) and (\ref{lem2:eq3}) together we have
	\[(y-\bar{y})(y+\bar{y}+1)=0. \]
	Since $y \ne \bar{y}$, we have
	\[y+\bar{y}+1=0. \]
	Plugging $\bar{y}=-y-1$ back into the left hand side of (\ref{lem2:eq2}), we obtain
	\begin{eqnarray} \label{lem2:eqy} y^2+y+1=0. \end{eqnarray}
	Now it is clear that if $q \equiv 2 \pmod{3}$, two distinct roots of the equation (\ref{lem2:eqy}) satisfy $y \in \gf(q^2) \setminus \gf(q)$, and if $q \not \equiv 2 \pmod{3}$, the roots $y$ of the above equation are all in $\gf(q)$, which shall be excluded when calculating $b_3$. This concludes the proof of Lemma \ref{lem2:p}.
\end{proof}

\begin{rmk} For $p=2$, the value $b_3$ was evaluated in \cite{DobF} (see also \cite{XiaLi}); for $p=3$, the value $b_3$ was essentially evaluated in \cite{XiaLi2}. It is easy to check that these results are consistent with Lemma \ref{lem2:p}.
\end{rmk}

If we assume now that we have computed the values $N_4$ and $N_5$, then we could prove Theorem \ref{1:main} easily.

\noindent {\bf Proof of Theorem \ref{1:main}} Let us assume the values of $N_4$ and $N_5$, which were given by Theorem \ref{N4:con} and Theorem \ref{N5:con} respectively. Then by Lemmas \ref{lem1:p} and \ref{lem2:p}, we can solve $N_0,\cdots,N_3$ in terms of $N_4$ and $N_5$ explicitly, and thus we can obtain the value distribution of $s(a)$ for $a \in \gf(q^2)$ explicitly, which is summarized below:
\begin{eqnarray*}
	s(a)=
	\left\{ \begin{array}{llll}
		-p^m-1&\mbox{occurs}&\frac{1}{2}p^{2m}-\frac{1}{3}p^m-\frac{1}{6}b_3p^m+N_4+4N_5 &\mbox{times} \\
		-1&\mbox{occurs}&-\frac{1}{2}p^m+\frac{1}{2}b_3p^m-4N_4-15N_5&\mbox{times}\\
		p^m-1&\mbox{occurs}&\frac{1}{2}p^{2m}+p^m-\frac{1}{2}b_3p^m+6N_4+20N_5&\mbox{times}\\
		2p^m-1&\mbox{occurs}&-\frac{1}{6}p^m+\frac{1}{6}b_3p^m-4N_4-10N_5&\mbox{times}\\
		3p^m-1&\mbox{occurs}&N_4&\mbox{times}\\
		4p^m-1&\mbox{occurs}&N_5&\mbox{times}.
	\end{array}\right.
\end{eqnarray*}
Here $b_3$ is given in Lemma \ref{lem2:p}.

Since $\gcd(d,q^2-1)=1$, we have $s(0)=-1$. Excluding this value $s(0)$ from the above value distribution table, we obtain the value distribution of $C_d(\tau)$. This proves Theorem \ref{1:main}. $\hfill \square$

Now we focus on $N_4$ and $N_5$. For $a \in \gf(q^2)$, $N(a) = i$ ($i \in \{4,5\}$) if and only if there are elements $z_1, \ldots, z_5 \in \Uq$ such that
\begin{itemize}
	\item[(1)] $\#\{z_1,\cdots,z_5\}=i, \quad i=4$ or $5$,
	
	\item[(2)] $z^5+\bar{a}z^3+az^2+1=(z+z_1)(z+z_2)(z+z_3)(z+z_4)(z+z_5)$.
\end{itemize}
This is equivalent to
\begin{eqnarray*} \label{iden1:p=2} \left\{\begin{array}{l}
		z_1,\cdots,z_5 \in \Uq, \quad
		\#\{z_1,\cdots,z_5\}=i, \\
		\sum_r z_r=0, \quad
		\prod_r z_r=1, \\
		\sum_{r<s}z_rz_s=\bar{a}.
	\end{array}
	\right.\end{eqnarray*}
Thus by setting
\[\fB_i:=\left\{\{z_1, \cdots, z_5\} \subset \Uq: \#\{z_1,\cdots,z_5\}=i, \sum_r z_r=0 \mbox{ and } \prod_r z_r=1\right\},\]
we have $N_i=\#\fB_i$ for $i \in \{4,5\}$. Here the symbol $\#A$ denotes the cardinality of any finite set $A$. Let us define
\[\fA_i:=\left\{\{z_1, \cdots, z_5\} \subset \Uq: \#\{z_1, \cdots,z_5\}=i \mbox{ and } \sum_r z_r=0 \right\}.\]
\begin{lemma} \label{lem3:p}
	\[N_i=\#\fB_i=\frac{\#\fA_i}{q+1}, \quad i \in \{4,5\}.\]
\end{lemma}
\begin{proof}
	For any $\underline{z}=\{z_1,\cdots,z_5\} \in \fA_i$ where $i=4$ or $5$, since $1=(d,q^2-1)=(5,q+1)$, there is a unique $c \in \Uq$ such that $\prod_i z_i=c^5$, so $\underline{z}/{c}:=\{\frac{z_1}{c},\cdots,\frac{z_5}{c}\} \in \fB_i$. This defines a map $\phi: \fA_i \to \fB_i$. It is clear that $\phi(\underline{x})=\phi(\underline{y})$ if and only if $\underline{x}=c \cdot \underline{y}$ for some $c \in \Uq$. So $\phi$ is a $(q+1)$-to-1 map, and Lemma \ref{lem3:p} is proved.
\end{proof}

For the remaining part of the paper, we will compute $N_4$ and $N_5$, focusing on the case that $p \ge 5$.

	\section{Computation of $N_4$}\label{sec:N4}
		
By Lemma \ref{lem3:p},
\begin{eqnarray*}
	N_4&=&\frac{1}{q+1}\#\left\{\{z_1, \cdots, z_5\} \subset \Uq: \#\{z_1, \cdots,z_5\}=4 \mbox{ and } \sum_r z_r=0 \right\}\\
	&=&\#\big\{\{x_1,x_2,x_3\}\subset \Uq \setminus\{1\}: \#\{x_1, x_2, x_3\}=3~\mathrm{and}~x_1+x_2+x_3+2=0\big\}
\end{eqnarray*}
Setting
\begin{eqnarray*} \label{4:x2}
	X'':=\left\{(x_1,x_2,x_3) \in \left(\Uq\setminus \{1\}\right)^3: x_1,x_2,x_3~\mbox{distinct and } x_1+x_2+x_3+2=0\right\},
\end{eqnarray*}
we have 
\begin{eqnarray*} \label{4:n4} N_4=\frac{\#X''}{6}. \end{eqnarray*}

Since $N_4$ has been obtained for $p=2$ in \cite{DobF} and for $p=3$ in \cite{XiaLi2} respectively, we assume now that $p \ge 5$. Define
\begin{eqnarray*} \label{n4:x}
	X:&=&\big\{(x_1,x_2,x_3)\in \Uq^3: x_1+x_2+x_3+2=0\big\},\\
	\triangle:&=& \bigcup_{i \ne j} X \cap \{x_i =x_j\}, \quad X'=X \setminus \triangle,\\
	H:&=& \bigcup_j X' \cap \{x_j=1\}.
\end{eqnarray*}
It is easy to see that
\[X''=X' \setminus H. \]
We will compute $\#X, \#\triangle$ and $\#H$ respectively to obtain
\[\#X''=\#X-\#\triangle-\#H. \]
The basic tools are the following two lemmas.
\begin{lemma}\cite{TZ}\label{TZ}
	Let $a,b\in\gf({q^2})^*$. If $b=a/\overline{a}$, then the roots of the quadratic equation $x^2+ax+b=0$ are all in $\Uq$, and the number of roots is $1-\eta({\frac{a\overline{a}-4}{a\overline{a}}})$. Here $\eta$ is the quadratic character on $\gf(q)$.
\end{lemma}

\begin{lemma}\label{map}	For any $a \in \gf(q)$, the number of roots $x \in \Uq$ such that $x+x^{-1}=a$ is $1-\eta(a^2-4)$.

\end{lemma}
\begin{proof} If $a=\pm 2$, it is easy to verify that $x+x^{-1}=a$ has a unique root $x=\frac{a}{2}$ in $\Uq$. Now suppose $a\in\gf(q) \setminus \{2,-2\}$ and suppose $x+x^{-1}=a$. This is equivalent to solving $x^2-ax+1=0$ for $x \in \Uq$. Since $a \ne \pm 2$, the discriminant is $a^2-4\in\gf(q)^*$. By Lemma \ref{TZ}, the equation has two distinct solutions $x \in \Uq$ if and only if $\eta\left(\frac{a^2-4}{a^2}\right)=\eta(a^2-4)=-1$, that is, the number of such $x$ is given by $1-\eta(a^2-4)$. This proves Lemma \ref{map}.
\end{proof}

Now we can prove	
\begin{lemma}\label{number} Let $q=p^m$ and $p \ge 5$. Then we have
	\begin{itemize}
		\item[(1)] $\#X=q+2+\lambda_{p^m}$,
		\item[(2)] $\#\triangle = 3-3\eta(-15)-2 \delta_{p,5}$,
		\item[(3)] $\#H=3 \left(1-\eta(5)-\delta_{p,5}\right)$,
		\item[(4)] $\#X''=q-4+3 \eta(-15)+3\eta(5)+5 \delta_{p,5}+\lambda_{p^m}$.
	\end{itemize}
	Here
	\[\delta_{a,b}=\left\{\begin{array}{lll}1&:& \mbox{ if } a =b,\\
		0&:& \mbox{ if } a \ne b, \end{array}\right.\]
	and $\lambda_{p^m}$ is given by Theorem \ref{pre:thm1}.
\end{lemma}
\begin{proof} (1) For any $x_1,x_2,x_3\in \Uq$ such that $x_1+x_2+x_3=2$, taking the $q$-th power on both sides of the equation, we obtain
	\[x_1^{-1}+x_2^{-1}+x_3^{-1}+2=0.\]
	So we have $x_1+x_2=-x_3-2$ and $x_1x_2=(x_1+x_2)/(x_1^{-1}+x_2^{-1})=\frac{x_3+2}{x_3^{-1}+2}$. We conclude that $x_1$ and $x_2$ are the two roots of the quadratic equation
	\begin{equation}\label{quadratic}
		t^2+(x_3+2)t+\frac{x_3+2}{x_3^{-1}+2}=0
	\end{equation}
	on the variable $t$. Since $x_3 \in \Uq$ and $p \ge 5$, by Lemma \ref{TZ}, the roots of the equation (\ref{quadratic}) are all in $\Uq$, and the number of roots is $1-\eta\left(\frac{(x_3+2)(x_3^{-1}+2)-4}{(x_3+2)(x_3^{-1}+2)}\right)=
	1-\eta\left(\frac{2(x_3+x_3^{-1})+1}{2(x_3+x_3^{-1})+5}\right)$. So we have
	\begin{align*}
		\#X&=\sum_{x_3\in \Uq}\left\{1-\eta\left(\frac{2(x_3+x_3^{-1})+1}{2(x_3+x_3^{-1})+5}\right)\right\}.
	\end{align*}
	The sum of the first term is $q+1$. As for the second term, taking $a=x_3+x_3^{-1}$, by Lemma \ref{map} we have
	\begin{align*}
		\#X&=q+1-\sum_{a\in\gf(q)}\left(1-\eta(a^2-4)\right) \cdot \eta\left((2a+1)(2a+5)\right)\\			 &=q+1-\sum_{a\in\gf(q)}\eta\left((2a+1)(2a+5)\right)+\sum_{a\in\gf(q)}\eta\left((a^2-4)(2a+1)(2a+5)\right).
	\end{align*}
	It is easy to see that the second character sum is $-1$ by Lemma \ref{charactersumquadratic}, and the third character sum is $\lambda_{p^m}$ given by (\ref{2:f1}), after substituting $2a$ by $a$. This proves the desired formula for $\#X$.
	
	(2)	We may define for $i \ne j$
	\[\triangle_{i,j}= X \cap \{x_i =x_j\},\]
	and thus \[\triangle=\triangle_{1,2} \cup \triangle_{1,3} \cup \triangle_{2,3}.\]
	For $\triangle_{1,2}$, following the argument in (1), we need to find the number of $x_3 \in \Uq$ such that Equation (\ref{quadratic}) has a double root (which were $x_1,x_2$) in $\Uq$, by Lemma \ref{map}, this requires $(x_3+2)(x_3^{-1}+2)=4$, that is, $x_3+x_3^{-1}=-\frac{1}{2}$, the number of such $x_3 \in \Uq$ by Lemma \ref{map} again is $1-\eta(-15)$. So $\triangle_{1,2}=1-\eta(-15)$. By symmetry we have
	\[\#\triangle_{i,j}=1-\eta(-15),\quad \forall i \ne j.\]
	It is trivial to see that \[\# \left(\triangle_{i,j} \cap \triangle_{i,k}\right)=\#\left\{x \in \Uq:3x+2=0\right\}=\delta_{p,5}, \quad i,j, k~\mbox{distinct}.\]
	The desired result for $\#\triangle$ is obtained by a simple application of the inclusion-exclusion principle.
	
	(3) The proof is also very simple, similar to the proof of (2), by using Lemma \ref{TZ} and Lemma \ref{map} and the inclusion-exclusion principle. We omit the details. Then (4) follows immediately as $\#X''=\#X-\#\triangle-\#H$.
\end{proof}

Since $N_4=\#X''/6$, by using Theorem \ref{pre:thm1} and Lemma \ref{number}, we can summarize the result of $N_4$ as follows. Note that we use Jacobi symbols $\left(\frac{\cdot}{\cdot}\right)$ here because for example $\left(\frac{5}{q}\right)=\eta(5)$.

\begin{theorem}\label{N4:con}When $p=5$,
	\[N_4=\frac{1}{6}\big(q-(-1)^m\big).\]
	When $p\geq 7$,
	\[N_4=\frac{1}{6}\left(q-4+3\left(\frac{5}{q}\right)+3\left(\frac{-15}{q}\right)+\lambda_{p^m}\right).\]
	Here the character sum $\lambda_{p^m}$ is given by Theorem \ref{pre:thm1}.
\end{theorem}

	\section{Melas codes, Zetterberg codes, and the ``pure weight'' distribution}\label{MZcodes}
As Lemma \ref{lem3:p} indicates, which was also hinted in \cite{XiaLi2}, the quantity $N_5$ is closely related to the number of ''pure weight'' 5 codewords in a $p$-ary Zetterberg code. This represents the most difficult part of the paper and is non-trivial even when $p=3$ which was successfully studied in \cite{XiaLi2}. The strategy we take is to transform this problem by MacWilliams identity via complete weight distribution to counting the number of codewords of weight at most 5 with various patterns in the corresponding Melas code. In this Section we will introduce basic information regarding Melas codes, Zetterberg codes, the complete weight distribution and how they are related to the quantity $N_5$. In the next Section, we will carry out the computation of $N_5$.

\subsection{Melas codes and Zetterberg codes}

Let $p$ be a prime and $q=p^m$ for some positive integer $m \ge 1$. Melas codes and Zetterberg codes can be defined over any finite field $\gf(q)$. However, for the case that is most interesting to us, we only consider such codes over the prime field $\gf(p)$.

Denote by $\tr_{q/p}: \gf(q) \to \gf(p)$ the standard trace function. The unit circle of $\gf(q^2)$ is defined by
\[\Uq=\{x \in \gf(q^2): x^{q+1}=x \bar{x}=1\}. \]
Here $\bar{x}:=x^q$ for any $x \in \gf(q^2)$. The $p$-ary Zetterberg code $Z(q)$ is a linear code of length $q+1$ over $\gf(p)$ defined as
\begin{eqnarray*} \label{2:zet}
	Z(q):=\left\{(c_x)_{x \in \Uq} \in \gf(p)^{q+1}: \sum_{x \in \Uq} c_x \cdot x=0\right\}.
\end{eqnarray*}
By Delsarte's theorem \cite{MacW}, the dual of the Zetterberg code is
\begin{eqnarray*} \label{2:zqp}
	Z(q)^{\bot}:=\left\{c(a)=\left(\tr_{q^2/p}(ax)\right)_{x \in \Uq}: a \in \gf(q^2)\right\}.
\end{eqnarray*}
The $p$-ary Melas code $M(q)$ is a linear code of length $q-1$ over $\gf(p)$ defined as
\begin{eqnarray*} \label{2:mel}
	M(q):=\left\{(c_x)_{x \in \gf(q)^*} \in \gf(p)^{q-1}: \begin{array}{lll}
		\sum_{x \in \gf(q)^*} c_x \cdot x&=&0\\
		\sum_{x \in \gf(q)^*}  c_x \cdot x^{-1}&=&0 \end{array}\right\}.
\end{eqnarray*}
By Delsarte's theorem \cite{MacW} again, the dual of the Melas code $M(q)$ is
\begin{eqnarray*} \label{2:mep}
	M(q)^{\bot}:=\left\{d(a,b)=\left(\tr_{q/p}(ax+bx^{-1})\right)_{x \in \gf(q)^*}: a,b \in \gf(q)\right\}.
\end{eqnarray*}
For simplicity by ordering the elements of $\gf(p)$ as $u_i=i$ for $0 \le i \le p-1$, we can define the complete weight distribution of linear codes $Z(q)^\bot$ and $M(q)^\bot$ as describe in Section \ref{pre}. It turns out that they are closely related with each other:

\begin{lemma} \label{2:lem} For $a \in \gf(q^2)^*$, let $b:=a \bar{a}$, then
	\begin{eqnarray} \label{2:comp-rel}
		\Comp(c(a))=\frac{2q}{p} \cdot \mathbf{1}-\Comp(d(b,1)).
	\end{eqnarray}
	Here $\mathbf{1}=(1,\ldots,1)$ is the vector with all entries being 1.
\end{lemma}

\begin{proof}
	For $0 \le i \le p-1$, let
	\begin{eqnarray*} t_i&=&\#\left\{u \in \Uq: \tr_{q^2/p}(au)=u_i\right\}\\
		&=&\#\left\{u \in \Uq: \tr_{q/p}(au+b (au)^{-1})=u_i\right\},\end{eqnarray*}
	\begin{eqnarray*} s_i&=&\#\left\{x \in \gf(q)^*: \tr_{q/p}(bx+x^{-1})=u_i\right\}\\
		&=&\#\left\{x \in \gf(q)^*: \tr_{q/p}(x+bx^{-1})=u_i\right\}. \end{eqnarray*}
	It was proved in Section 3 of \cite{vander2} that $t_0+s_0=2p^{m-1}$, that is, for $u_0=0$ (see Equation (*) on page 265 and proofs on page 266 in \cite{vander2}). The proofs depend on properties of the maps $\alpha(x)=x+bx^{-1}$ for $x \in \gf(q)^*$ and $\beta(u)=au+b(au)^{-1}$ for $u \in \Uq$. It is easy to note that the proofs work for any $i$, that is, $t_i+s_i=2p^{m-1}$ for any $i$. This yields the desired relation (\ref{2:comp-rel}).
\end{proof}

\begin{lemma} \label{2:lem2}
	Let $W_{Z(q)^\bot}(\bz)$ and $W_{M(q)^\bot}(\bz)$ be the complete weight enumerators of $Z(q)^\bot$ and $M(q)^\bot$ respectively. Then we have
	\begin{eqnarray} \label{2:zmqC}
		W_{Z(q)^\bot}(\bz)=z_0^{q+1}-2(q+1)z_0 \bz^{\frac{q}{p} \cdot \bone} +\frac{q+1}{q-1}\bz^{\frac{2q}{p} \cdot \bone} \left(W_{M(q)^\bot}(\bz^{-1})-z_0^{1-q}\right).
	\end{eqnarray}
	
\end{lemma}
\begin{proof} Let $\alpha$ be a generator of $\gf(q^2)^*$. Then $\beta=\alpha \bar{\alpha}$ is a generator of $\gf(q)^*$.
	
	First for the complete weight distribution of $Z(q)^\bot$, it suffices to consider the complete weight $\Comp(c(\alpha^j))$ of the words
	\[c(\alpha^j)=\left(\tr_{q^2/p}(\alpha^j x)\right)_{x \in \Uq} \qquad \mbox{ for } 0 \le j \le q-2. \]
	We have to count them with multiplicity $q+1$ since multiplying $\alpha^j$ by an element of $\Uq$ permutes the coordinates of a code word. Thus we have
	\begin{eqnarray} \label{2:zqC}
		W_{Z(q)^\bot}(\bz)=z_0^{q+1}+(q+1) \sum_{j=0}^{q-2} \bz^{\Comp(c(\alpha^j))}.
	\end{eqnarray}
	Next for the complete weight distribution of $M(q)^\bot$, noting that for $a,b \in \gf(q)$, a codeword of $d(a,b)$ is given by
	\[d(a,b)=\left(ax+bx^{-1}\right)_{x \in \gf(q)^*}. \]
	If $a \ne 0$, it is easy to see that
	\[\Comp(d(a,0))=\Comp(d(0,a))=\Comp(d(1,0))=\left(\frac{q}{p}-1,\frac{q}{p}, \cdots, \frac{q}{p}\right). \]
	This term corresponds to $z_0^{-1} \bz^{\frac{q}{p} \cdot \bone}$ and shall be counted $2(q-1)$ times in $W_{M(q)^\bot}(\bz)$.
	Next, it is easy to see that if $ab \ne 0$ then
	\[\Comp(d(a,b))=\Comp(d(ab,1)). \]
	Since $\beta$ is a generator of $\gf(q)^*$, each term $\Comp(d(\beta^j,1))$ shall be counted with multiplicity $q-1$. Thus we have
	\begin{eqnarray} \label{2:mqC}
		W_{M(q)^\bot}(\bz)=z_0^{q-1}+2(q-1)z_0^{-1} \bz^{\frac{q}{p} \cdot \bone}+(q-1) \sum_{j=0}^{q-2} \bz^{\Comp(d(\beta^j,1))}.
	\end{eqnarray}
	Noting from (\ref{2:comp-rel}) that
	\[\bz^{\Comp(c(\alpha^j)}=\bz^{\frac{2q}{p}\cdot \bone} \, \bz^{-\Comp(d(\beta^j,1))}, \]
	and combining (\ref{2:zqC}) and (\ref{2:mqC}), we obtain the desired identity (\ref{2:zmqC}).
\end{proof}

Now by Lemma \ref{2:lem2} and by applying the MacWilliams identity (\ref{2:MacW}) twice, noting that $Z(q)^\bot$ is a $[q+1,2m]_p$ code and $M(q)^\bot$ is a $[q-1,2m]_p$ code, we obtain the relation between the complete enumerators of $Z(q)$ and $M(q)$ as follows:

\begin{theorem} \label{2:thm2} Let $W_{Z(q)}(\bz)$ and $W_{M(q)}(\bz)$ be the complete weight enumerators of $Z(q)$ and $M(q)$ respectively. Then we have
	\begin{eqnarray} \label{2:ZMC}
		q^2W_{Z(q)}(\bz)=x_0^{q+1}-2(q+1)x_0 \bx^{\frac{q}{p} \cdot \bone} +\frac{q+1}{q-1}\bx^{\frac{2q}{p} \cdot \bone} \left(\frac{1}{p^{q-1-2m}}W_{M(q)}(\by)-x_0^{1-q}\right).
	\end{eqnarray}
	Here for the multi-variables $\bz=(z_0,\ldots,z_{p-1}), \bx=(x_0,\ldots,x_{p-1}), \by=(y_0,\ldots,y_{p-1})$, the relations are
	\begin{eqnarray} \label{thm2:xyz}
		x_i=\sum_{s=0}^{p-1} \psi(u_iu_s) z_s, \quad y_i=\sum_{s=0}^{p-1} \psi(u_iu_s) x_s^{-1} \quad \mbox{ for any }i=0, 1, \cdots, p-1,
	\end{eqnarray}
	and $\psi: \gf(p) \to \mathbb{C}^\times$ is the standard additive character given by $x \mapsto \omega_p^x$ where $\omega_p=\exp(2 \pi \sqrt{-1}/p)$.
\end{theorem}

\subsection{The pure weight distribution of Zetterberg codes}

Recall that we have ordered elements of $\gf(p)$ as $u_i=i$ for $0 \le i \le p-1$. Denote by $W(z)$ the function in variable $z$ obtained by setting $z_0=1, z_1=\ldots=z_{p-2}=0, z_{p-1}=-z$ in the complete weight enumerator $W_{Z(q)}(\bz)$,
that is,
\begin{eqnarray} \label{2:pw} W(z)=W_{Z(q)}(1,0,\ldots,0,-z)=\sum_{t=0}^{q+1} (-1)^t B_t z^t.\end{eqnarray}
We may call $W(z)$ as a ``pure weight enumerator'' of $Z(q)$ because
\begin{eqnarray} \label{b:bi} B_i=\#\left\{\bc \in Z(q): \Comp(\bc)=(q+1-i,0,\ldots,0,i)\right\}, \end{eqnarray}
that is, $B_i$ counts the number of ``pure weight'' $i$ codewords of $Z(q)$ (with respect to the symbol $p-1$). Then we have
\begin{theorem} \label{2:thm3}
	\begin{eqnarray*} \label{thm3:weight}
		q^2 W(z)&=&(1-z)^{q+1}-2(q+1)(1-z)\left(1-z^p\right)^{\frac{q}{p}}-\frac{q+1}{q-1}\left(1-z^p\right)^{\frac{2q}{p}}(1-z)^{1-q} \\
		&&+\, \frac{q^2(q+1)}{q-1} \left(1-z^p\right)^{\frac{2q}{p}-q+1}W_{M(q)}(\v),
	\end{eqnarray*}
	where $W_{M(q)}(\v)$ is the complete weight enumerator of $M(q)$ and $\v=(v_0,\ldots,v_{p-1})$ is given by
	\begin{eqnarray*} \label{2:thm3-y} v_i=z^i, \quad 0 \le i \le p-1. \end{eqnarray*}
\end{theorem}
\begin{proof}
	We simply use $u_i=i$ ($0 \le i \le p-1$), $z_0=1, z_1=\cdots=z_{p-2}=0, z_{p-1}=-z$ and simplify the right hand side of (\ref{2:ZMC}) in Theorem \ref{2:thm2}. First, from (\ref{thm2:xyz}) we have
	\[x_i=1-\psi(-i)z, \quad y_i=\sum_{s=0}^{p-1} \frac{\psi(is)}{1-\psi(-s)z} \quad \mbox{ for } 0 \le i \le p-1. \]
	We may expand $y_i$ as
	\begin{eqnarray*}
		y_i= \sum_{s=0}^{p-1}\psi(is) \sum_{t \ge 0} \left(\psi(-s)z\right)^t=\sum_{t \ge 0}z^t\sum_{s=0}^{p-1}\psi(s(i-t)).
	\end{eqnarray*}
	Since the subscript $s$ runs over $\gf(p)$ and $\psi: \gf(p) \to \mathbb{C}^\times$ is the standard additive character, the inner sum is $0$ unless $t \equiv i \pmod{p}$, thus we obtain
	\begin{eqnarray*}
		y_i= \sum_{\substack{t \ge 0\\
				t \equiv i \pmod{p}}}z^t \cdot p=p\cdot z^i \sum_{k \ge 0}z^{kp}=\frac{pv_i}{1-z^p}, \quad 0 \le i \le p-1.
	\end{eqnarray*}
	Using the above relation between $y_i$ and $v_i$, we find
	\[W_{M(q)}(\by)=\left(\frac{p}{1-z^p}\right)^{q-1} W_{M(q)}(\v). \]
	Finally, using $x_0=1-z$ and
	\begin{eqnarray*}
		\bx^{\bone}&=&\prod_{i=0}^{p-1}x_i=\prod_{i=0}^{p-1}\left(1-\psi(-i)z\right)\\
		&=& \prod_{i=0}^{p-1}\left(1-\zeta_p^{-i}z\right)=1-z^p,
	\end{eqnarray*}
	we can simplify the right hand side of Equation (\ref{2:ZMC}) to yield the desired result in Theorem \ref{2:thm3}.
\end{proof}

\section{Computation of $N_5$}\label{sec:N5}
Now we resume our computation of $N_5$. We have observed from Lemma \ref{lem3:p} that
\begin{eqnarray*} \label{n5:b} N_5=B_5/(q+1),
\end{eqnarray*}
where $B_5=\#\fA_5$ is the number of codewords of ``pure weight'' five of $Z(q)$. Noting from (\ref{2:pw}) that $-B_5$ is the coefficient of $z^5$ in the pure weight enumerator $W(z)$ of the Zetterberg code $Z(q)$, which in turn is related to the complete weight enumerator $W_{M(q)}(\bu)$ of the Melas codes $M(q)$ by Theorem \ref{2:thm3}.

Let us use some notation. Denote by $W_{M(q)}(\bz)$ be the complete weight enumerator of $M(q)$ given by
\begin{eqnarray*} \label{3:ww}
	W_{M(q)}(\bz)=\sum_{\bt} A_{\bt}\bz^{\bt}.
\end{eqnarray*}
For $\bt=(t_0,t_1,\ldots,t_{p-1})$, we may write it simply as $\bt=1^{t_1}\cdots (p-1)^{t_{p-1}}$ and omits those $t_i$'s which are zero. For example, $\bt=\mathbf{1^23^1}$ means that $t_1=2,t_3=1$ and $t_i=0$ for $i \ne 1,3$.

So by comparing the coefficients of $z^5$ on both sides of the equation in Theorem \ref{2:thm3}, we can obtain $B_5$ explicitly in terms of some complete weights in the Melas code.

\begin{lemma} \label{4:lem} Let $(A_{\bt})_{\bt}$ be the complete weight distribution of the Melas code $M(q)$, and define
	\[\Gamma_d:=\sum_{\substack{\bt=(t_0,t_1,\ldots,t_{p-1}) \in \Sigma^{q-1,p}\\
			\sum_i it_i=d}} A_{\bt}. \]
	Then for $p \ge 3$ we have
	\begin{eqnarray*} \label{4:b5}
		B_5&=&\frac{(q+1)(q^2+11)}{60}-\frac{q+1}{q-1} \, \Gamma_5-\delta_{p,5}\frac{3(q+1)}{5}-\delta_{p,3} \frac{q^2-1}{6}\\
		&&-\delta_{p,3}(q^2-1)\left(1-\frac{2q}{p(q-1)}\right)\left(\frac{\Gamma_2}{q-1}-\frac{1}{2}\right).
	\end{eqnarray*}
\end{lemma}
\begin{proof}
	The coefficient of $z^5$ on the left side of Theorem \ref{2:thm3} is $(-1)^5q^2B_5$. Denote by $T$ the coefficient of $z^5$ on the right side of Theorem \ref{2:thm3}. Using
	\[(1-z^p)^{\frac{q}{p}}=1-\frac{q}{p}z^p+O(z^{2p}), \quad (1-z^p)^{\frac{2q}{p}}=1-\frac{2q}{p}z^p+O(z^{2p}),\]
	\[(1-z)^{q+1}=\sum_i (-1)^k \binom{q+1}{k}z^k, \quad (1-z)^{1-q}=\sum_i \binom{q+k-2}{k}z^k,\]
	\begin{eqnarray*} W_{M(q)}(\bu)&=&\sum_{\bt}A_{\bt}\,u_0^{t_0}u_1^{t_1} \cdots u_{p-1}^{p-1}=\sum_{\bt} A_{\bt} \,z^{\sum_i it_i}=\sum_d \Gamma_d z^d,\end{eqnarray*}
	we can obtain
	\begin{eqnarray*} T&=&(-1)^5 \binom{q+1}{5}-2(q+1)\delta_{p,5}\left(-\frac{q}{p}\right)\\
		&&- \frac{q+1}{q-1}\left(\binom{q+5-2}{5}+\delta_{p,5} \left(-\frac{2q}{p}\right)+\delta_{p,3} \left(-\frac{2q}{p}\right)\binom{q+2-2}{2}\right)\\
		&&+\frac{q^2(q+1)}{q-1}\left(\Gamma_5+\delta_{p,5}\left(q-\frac{2q}{p}-1\right)+
		\delta_{p,3}\left(q-\frac{2q}{p}-1\right)\Gamma_2\right).\end{eqnarray*}
	Now use $B_5=-\frac{T}{q^2}$ and do some simplification on $T$, we obtain the desired formula for $B_5$.
\end{proof}

\begin{lemma} \label{4:gamma25} We have
	\begin{eqnarray*} \label{4:gamma2} \Gamma_2=\frac{q-1}{2},
	\end{eqnarray*}
	and
	\begin{eqnarray} \label{4:gamma5}
		\Gamma_5=\left\{\begin{array}{ll}
			(q-1)\left(q^2+q\left(-14+(-1)^m\right)+36\right)/120 &\mbox{ if } p =3,\\
			(q-1)\left(q^2+q\left(7+4(-1)^m\right)-50-10(-1)^m\right)/120 & \mbox{ if } p =5,\\
			(q-1)\left(q^2+q \left(6+4\left(\frac{q}{3}\right) \right)+6+20\left(\frac{5}{q}\right)+15\left(\frac{-15}{q}\right)+10 \lambda_{q}+A_q\right)/120 & \mbox{ if } p \ge 7.
		\end{array}
		\right.
	\end{eqnarray}
\end{lemma}
\begin{proof}
	See Section \ref{appen} {\bf Appendix}.
\end{proof}
Now plugging these values $\Gamma_2, \Gamma_4$ into the right hand side of the equation for $B_5$ in Lemma \ref{4:lem}, and noting $N_5=B_5/(q+1)$, we can finally obtain the value for $N_5$, which we summarise below.

\begin{theorem} \label{N5:con}
	\begin{eqnarray*} N_5=\left\{\begin{array}{ll}
			\left(q^2-q\left(6+(-1)^m\right)+6\right)/120 &\mbox{ if } p =3,\\
			\left(q^2-q\left(7+4(-1)^m\right)+10(-1)^m\right)/120 & \mbox{ if } p =5,\\
			\left(q^2-q \left(6+4\left(\frac{q}{3}\right) \right)+16-20\left(\frac{5}{q}\right)-15\left(\frac{-15}{q}\right)-10 \lambda_{q}-A_q\right)/120 & \mbox{ if } p \ge 7.
		\end{array}
		\right.
	\end{eqnarray*}
	Here $\left(\frac{\cdot}{\cdot}\right)$ is the Jacobi symbol, the value $A_q$ is given in Theorem \ref{5:nqxb} and $\lambda_{q}$ is given in Theorem \ref{pre:thm1}.
\end{theorem}
\begin{rmk} For $p=3$, $N_5$ has been computed (see \cite[Lemma 9]{XiaLi2}).
\end{rmk}

	\section{concluding remarks} \label{sec:cond}
In this paper we computed the cross-correlation distribution of the Niho-type decimation $d=3(p^m-1)+1$ over $\gf(p^{2m})$ for any prime $p \ge 5$ where $\gcd(d,p^{2m}-1)=1$. The main difficulty of this problem for $p \ge 5$ is to count the number of codewords of ``pure weight'' 5 in $p$-ary Zetterberg codes. We solve this problem by using the complete weight distribution of Zetterberg codes and MacWilliams identity to convert it into counting the number of codewords of small weight of various patterns in $p$-ary Melas codes, the most difficult of which is related to a K3 surface well studied in the literature and hence can be computed.

	\section{Appendix: counting codewords of small weight in Melas codes} \label{appen}
	
	In order to prove Lemma \ref{4:gamma25}, noting that 
	\[\Gamma_2=A_{\mathbf{1^2}}+A_{\mathbf{2^1}},\]
	\[\Gamma_5=A_{\mathbf{5^1}}+A_{\mathbf{1^14^1}}+A_{\mathbf{2^13^1}}+ A_{\mathbf{1^23^1}}+A_{\mathbf{1^32^1}}+A_{\mathbf{1^12^2}}+A_{\mathbf{1^5}},\] 
	we need to compute the quantities $A_{\mathbf{1^2}}, A_{\mathbf{2^1}},A_{\mathbf{5^1}},A_{\mathbf{1^14^1}}, A_{\mathbf{2^13^1}}, A_{\mathbf{1^23^1}}, A_{\mathbf{1^32^1}},A_{\mathbf{1^12^2}},A_{\mathbf{1^5}}$ first. These quantities correspond to the number of codewords of weight at most 5 of certain patterns in the Melas code $M(q)$. 
	
	\subsection{Codewords of weight at most two in Melas codes}
	\begin{lemma} \label{pre:N1} For any $p \ge 2$ we have 
		\begin{itemize} 
			\item[(i)] $A_{\mathbf{2^1}}=A_{\mathbf{5^1}}=0$,
			\item[(ii)] $A_{\mathbf{1^2}}=\frac{q-1}{2}$,
			\item[(iii)] $A_{\mathbf{1^14^1}}=A_{\mathbf{2^13^1}}=0$.
		\end{itemize}
	\end{lemma}
	\begin{proof} (i) This is obvious. For example, the value $A_{\mathbf{2^1}}$ counts the number of weight $1$ codewords (with respect to the symbol $2$) in $M(q)$, the number of which is obviously $0$. 
	
	(ii) The value $A_{\mathbf{1^2}}$ counts the number of ``pure weight'' 2 codewords of $M(q)$, or equivalently, $A_{\mathbf{1^2}}$ is half the number of solutions $x_1, x_2 \in \gf(q)^*$ such that
		\[x_1+x_2=0, \quad x_1^{-1}+x_2^{-1}=0, \quad x_1 \ne x_2. \]
		Since the only solution is $x_1=-x_2$, we see that $A_{\mathbf{1^2}}=\frac{q-1}{2}$.
		
		(iii) The existence of the value $A_{\mathbf{1^14^1}}$ requires that $p \ge 5$, and $A_{\mathbf{1^14^1}}$ is half the number of solutions $x_1,x_2 \in \gf(q)^*$ such that 
		\[x_1+4x_2=0, \quad x_1^{-1}+4x_2^{-1}=0, \quad x_1 \ne x_2. \]
		It is easy to conclude from these equations that $A_{\mathbf{1^14^1}}=0$. The proof for $A_{\mathbf{2^13^1}}$ is the same. 
	\end{proof}

	\subsection{Codewords of weight 3 in Melas codes}
	\begin{lemma} \label{pre:N4}
		Denote by $N_{\mathbf{1^23^1}}$ the number of solution $(x_1,x_2,x_3) \in (\gf(q)^*)^3$ such that
		\begin{eqnarray*} \label{pre:N1^23^1}
			\left\{\begin{array}{lll}
				x_1+x_2+3x_3&=&0,\\
				\frac{1}{x_1}+\frac{1}{x_2}+\frac{3}{x_3}&=&0, \end{array}\right. \quad x_1,x_2,x_3 \emph{\mbox{ distinct}}.
		\end{eqnarray*}
		Then
		\begin{eqnarray*} \label{N1^23^1}
			N_{\mathbf{1^23^1}}=\left\{ \begin{array}{ll}
				(q-1)\left(1+(-1)^m \right) &\mathrm{if} ~~p=2, \\
				(q-1)(q-3) &\mathrm{if}  ~~p=3,\\
				0 &\mathrm{if}  ~~p=5,\\
				(q-1)\left(1+\left(\frac{5}{q}\right) \right) &\mathrm{if}  ~~p\geq 7,
			\end{array}\right.
		\end{eqnarray*}
		and
		\[ A_{\mathbf{1^23^1}}=\left\{\begin{array}{ll}
			0&\mathrm{if} ~~ p=2,3,5,\\
			\frac{1}{2}\,N_{\mathbf{1^23^1}}& \mathrm{if} ~~p \geq 7.
		\end{array}\right.\]
		Here $\left(\frac{\cdot}{\cdot}\right)$ is the Jacobi symbol.
	\end{lemma}
	\begin{proof}
		If $p=3$, then we obtain $x_2=-x_1$. This with $x_3\neq0,\pm x_1$ leads to $N_{\mathbf{1^23^1}}=(q-3)(q-1)$.
		
		Now we assume that $p \ne 3$. Let $y_1=\frac{x_1}{x_3}$ and  $y_2=\frac{x_2}{x_3}$. Then $y_1,y_2\in\gf(q)\setminus\{0,1\}$, $y_1\neq y_2$, and they satisfy
		\begin{equation}\label{113'}
			\left\{ \begin{array}{ll}
				{y_1} + {y_2} + 3 &= 0\\
				\frac{1}{y_1}+\frac{1}{y_2}+3 &= 0.
			\end{array} \right.
		\end{equation}	
		We obtain $y_1y_2=1$ from (\ref{113'}). Then $y_1$ and $y_2$ are the two roots of the quadratic equation $t^2+3t+1=0$. When $p=2$, this equation is solvable in $\gf(q)$ if and only if $m$ is even with $(y_1,y_2)=(\omega, \omega^2)$ or $(y_1,y_2)=(\omega^2,\omega)$ where $\omega$ is a third-root of unity in $\gf(q)$. Note that in this case $y_1,y_2 \ne 1$. Considering the choices of $x_3$, we obtain $N_{\mathbf{1^23^2}}=(1+(-1)^m)(q-1)$; When $p=5$, we have $y_1=y_2=1$, which is a contradiction. Hence $N_{\mathbf{1^23^1}}=0$ when $p=5$; When $p\geq 7$, the equation is solvable in $\gf(q)$ if and only if $\eta(5)=1$, and we obtain  $(y_1,y_2)=(\frac{-3+\sqrt{5}}{2},\frac{-3-\sqrt{5}}{2})$ or $(y_1,y_2)=(\frac{-3-\sqrt{5}}{2},\frac{-3+\sqrt{5}}{2})$. Note that in this case $y_1,y_2\neq 1$. Considering the choices of $x_3$, we obtain $N_{\mathbf{1^23^1}}=(1+\eta(5))(q-1)$. Finally noting that $5 \in \gf(p)$ and $\eta(5)=\left(\frac{5}{q}\right)$. The proof is completed.
	\end{proof}
	
	\begin{lemma} \label{pre:N5}
		Denote by $N_{\mathbf{1^12^2}}$ the number of solution $(x_1,x_2,x_3) \in (\gf(q)^*)^3$ such that
		\begin{eqnarray*} \label{pre:N1^12^2}
			\left\{\begin{array}{cll}
				x_1+2x_2+2x_3&=&0,\\
				\frac{1}{x_1}+\frac{2}{x_2}+\frac{2}{x_3}&=&0, \end{array}\right. \quad x_1,x_2,x_3 \emph{\mbox{ distinct}}.
		\end{eqnarray*}
		Then
		\begin{eqnarray*}
			N_{\mathbf{1^12^2}}=\left\{ \begin{array}{ll}
				0 &\mathrm{if}  ~~p=2,3,5,\\
				(q-1) \left(1+\left(\frac{-15}{q}\right) \right) &\mathrm{if}  ~~p\geq 7,
			\end{array} \right.
		\end{eqnarray*}
		and
		\[ A_{\mathbf{1^12^2}}=\left\{\begin{array}{ll}
			0&\mathrm{if} ~~ p=2,3,5,\\
			\frac{1}{2}\,N_{\mathbf{1^12^2}}& \mathrm{if} ~~p \geq 7.
		\end{array}\right.
		\]
	\end{lemma}
	\begin{proof} The proof is very similar to that of $N_{\mathbf{1^23^1}}$ so we omit it.
	\end{proof}
	\subsection{Codewords of weight 4 in Melas codes}
	\begin{lemma} \label{pre:N6}
		Denote by $N_{\mathbf{1^32^1}}$ the number of solution $(x_1,x_2,x_3,x_4) \in (\gf(q)^*)^4$ such that
		\begin{eqnarray*} \label{pre:N1^32^1}
			\left\{\begin{array}{lll}
				x_1+x_2+x_3+2x_4&=&0,\\
				\frac{1}{x_1}+\frac{1}{x_2}+\frac{1}{x_3}+\frac{2}{x_4}&=&0, \end{array}\right. \quad x_1,x_2,x_3,x_4 \emph{\mbox{ distinct}}.
		\end{eqnarray*}
		Then
		\begin{eqnarray*} \label{N1^32^1}
			N_{\mathbf{1^32^1}}=\left\{ \begin{array}{ll}
				(q-1)\left(1+(-1)^m\right)(q-4) &\mathrm{if} ~~ p=2,\\	
				0 &\mathrm{if}  ~~p=3,\\
				(q-1) \left(q-6-(-1)^m \right) &\mathrm{if}  ~~p=5,\\
				(q-1) \left( q-10-3\left(\frac{5}{q}\right)-3\left(\frac{-15}{q}\right)+\lambda_{q} \right) &\mathrm{if}  ~~p\geq 7,
			\end{array} \right.
		\end{eqnarray*}
		and
		\[ A_{\mathbf{1^32^1}}=\left\{\begin{array}{ll}
			0&\mathrm{if} ~~ p=2,3,\\
			\frac{1}{6}\,N_{\mathbf{1^32^1}}& \mathrm{if} ~~p \geq 5.
		\end{array}\right.\\
		\]
		Here $\lambda_{q}$ is defined in (\ref{2:f1}) and can be evaluated by Theorem \ref{pre:thm1}.
	\end{lemma}
	
	\begin{proof}
		Let $y_i=\frac{x_i}{x_4}$, $i=1,2,3$ and define
		\begin{equation*}\label{1112'}
			Y'':=\left\{ (y_1,y_2,y_3)\in \left(\gf(q)^*\setminus \{1\}\right)^3: \begin{array}{l}
y_1,y_2,y_3~\mbox{distinct }\\
\sum_i y_i+2=0, \sum_i y_i^{-1}+2=0.
\end{array}\right\}.
		\end{equation*}
We have
\begin{eqnarray} \label{ap:nyc} N_{\mathbf{1^32^1}}=(q-1) \cdot \#Y''. \end{eqnarray}
It is easy to verify that
\begin{eqnarray*}
			\#Y''= \left\{\begin{array}{cll}
				(1+(-1)^m)(q-4)&:&p=2,\\
				0&:&p=3,
\end{array}\right.
		\end{eqnarray*}
and we omit the details. Here we only treat the case that $p \ge 5$. The proof is very similar to that of Lemma \ref{number}. Define
\begin{eqnarray*} \label{n4:x}
Y:&=&\bigg\{(y_1,y_2,y_3)\in \left(\gf(q)^*\right)^3: \sum_i y_i+2=0, \sum_i y_i^{-1}+2=0\bigg\},\\
\triangle:&=& \bigcup_{i \ne j} Y \cap \{y_i =y_j\}, \quad Y'=Y \setminus \triangle,\\
H:&=& \bigcup_j Y' \cap \{y_j=1\}.
\end{eqnarray*}
Thus
\[Y''=Y' \setminus H. \]
We will compute $\#Y, \#\triangle$ and $\#H$ respectively to obtain
\[\#Y''=\#Y-\#\triangle-\#H. \]
We first compute $\#Y$, that is, we consider $(y_1,y_2,y_3) \in \left(\gf(q)^*\right)^3$ such that
\begin{eqnarray} \label{ap:yi} y_1+y_2=-2-y_3, \quad y_1^{-1}+y_2^{-1}=-2-y_3^{-1}.\end{eqnarray}
If $y_1+y_2=0$, then $y_1^{-1}+y_2^{-1}=\frac{y_1+y_2}{y_1y_2}=0$, hence $-2-y_3=-2-y_3^{-1}=0$, that is, $y_3=-2=-\frac{1}{2}$, which is impossible since $p \ge 5$. So $y_1+y_2 \ne 0$ and $y_1^{-1}+y_2^{-1} \ne 0$, so we have $y_3 \in \gf(q)^* \setminus \{-2,-\frac{1}{2}\}$.
Equation (\ref{ap:yi}) implies that
\[ y_1+y_2=-y_3-2, \quad y_1y_2=
		\frac{y_3+2}{y^{-1}_3+2}.\]
So $y_1,y_2 \in \gf(q)^*$ are the two roots of the quadratic equation
		\begin{equation}\label{n2quadratic}
			t^2+(y_3+2)t+\frac{y_3+2}{y^{-1}_3+2}=0.
		\end{equation}
Since $\frac{y_3+2}{y^{-1}_3+2} \ne 0$ as $y_3 \in \gf(q)^* \setminus \{-2,-\frac{1}{2}\}$, $t=0$ is not a root of Equation (\ref{n2quadratic}). The discriminant is
\[D=(y_3+2)^2-\frac{4(y_3+2)}{y^{-1}_3+2}=\frac{\left(2(y_3+y_3^{-1})+5\right) \cdot \left(2(y_3+y_3^{-1})+1\right)}{(y_3^{-1}+2)^2}.\]
So for a given $y_3$, the number of $y_1,y_2 \in \gf(q)^*$ satisfying (\ref{ap:yi}), which is the number of roots of Equation (\ref{n2quadratic}) in $\gf(q)^*$, is given by $1+\eta(D)$, so we have
\begin{align*}
			 \#Y&=\sum_{y_3\in\gf(q)^*\setminus\{-2,-\frac{1}{2}\}}\left(1+\eta(D)\right)\\
		&=q-3+\sum_{y_3\in\gf(q)^*\setminus\{-2,-\frac{1}{2}\}}\eta\left(\big(2(y_3+y_3^{-1})+5\big) \cdot \big(2(y_3+y_3^{-1})+1\big)\right)\\
&=q-3+\sum_{y_3\in\gf(q)^*}\eta\left(\big(2(y_3+y_3^{-1})+5\big) \cdot \big(2(y_3+y_3^{-1})+1\big)\right). 		 \end{align*}
Let $y_3+y_3^{-1}=a$, that is, $y_3^2-ay_3+1=0$. Easy to see that for any given $a \in \gf(q)$, the number of $y_3 \in \gf(q)$ satisfying the relation $y_3+y_3^{-1}=a$ is given by $1+\eta(a^2-4)$, so we have 		 \begin{align*}
			 \#Y&=q-3+\sum_{a \in\gf(q)}\left(1+\eta(a^2-4)\right) \cdot \eta\left((2a+5) (2a+1)\right)\\
&=q-3+\sum_{a \in\gf(q)} \eta\left((2a+5) (2a+1)\right)+\sum_{a \in\gf(q)}\eta\left((a^2-4)(2a+5) (2a+1)\right).		 \end{align*}
The second character sum is $-\eta(4)=-1$ by Lemma \ref{charactersumquadratic}, and the third character sum is $\lambda_{p^m}$, after substituting $2a$ by $a$. Therefore we obtain
\begin{eqnarray*} \label{ap:cY}
\#Y=q-4+\lambda_{p^m}.
\end{eqnarray*}
As for $\triangle$, denoting
\[\triangle_{i,j}=Y \cap \{y_i=y_j\}, \quad i<j. \]		
Very similar to the proof of Lemma \ref{number}, we can obtain
\[\#\triangle_{i,j}=1+\eta(-15), \quad \forall i<j,\]
\[\# \left(\triangle_{i,j} \cap \triangle_{i,k}\right)=\#\left\{y \in \gf(q)^*:3y+2=0,3y^{-1}+2=0\right\}=\delta_{p,5}, \quad i,j, k~\mbox{distinct}.\]
Noting
\[\triangle=\triangle_{1,2} \cup \triangle_{1,3} \cup \triangle_{2,3},\]
by a simple application of the inclusion-exclusion principle, we can obtain
\begin{eqnarray*} \label{ap:ctri}
\#\triangle=3+3 \eta(-15)-2 \delta_{p,5}.
\end{eqnarray*}
Finally for $H$, using similar argument in the proof of Lemma \ref{number}, by considering
\[H_i=Y' \cap \{y_i=1\}, \quad \forall i \]
for $i=1,2,3$ individually and by a simple application of the inclusion-exclusion principle again, we can obtain
\begin{eqnarray*} \label{ap:cH}
\#H=3\left(1+\eta(5)-\delta_{p,5}\right).
\end{eqnarray*}
We conclude that
\begin{eqnarray*} \#Y''&=&\#Y-\#\triangle-\#H\\
&=&q-10-3\eta(5)-3\eta(-15)+5\delta_{p,5}+\lambda_{p^m}. \end{eqnarray*}
Noting (\ref{ap:nyc}), we completes the proof of Lemma \ref{pre:N6}.
	\end{proof}
	
	\subsection{Codewords of ``pure weight'' 5 in the Melas codes and a K3 surface}
	
	Denote by $N_{\mathbf{1^5}}$ the number of solution $(x_1,x_2,x_3,x_4,x_5) \in (\gf(q)^*)^5$ such that
	\begin{eqnarray*} \label{pre:N1^23^1}
		\left\{\begin{array}{lll}
			x_1+x_2+x_3+x_4+x_5&=&0,\\
			 \frac{1}{x_1}+\frac{1}{x_2}+\frac{1}{x_3}+\frac{1}{x_4}+\frac{1}{x_5}&=&0, \end{array}\right. \quad x_1,x_2,x_3,x_4,x_5 \emph{\mbox{ distinct}}.
	\end{eqnarray*}
	The quantity $N_{\mathbf{1^5}}$ is closely related to the number of codewords of ``pure weight'' 5 in the Melas codes. It turns out $N_{\mathbf{1^5}}$ can be computed as a consequence of the main results in \cite{Peter} for any $p$, though in the paper only the case $p=2$ was treated explicitly (see \cite[Section 1. Melas codes]{Peter}). Here we follow the presentation of \cite{Peter} closely and give a detailed account of how to compute $N_{\mathbf{1^5}}$ for any $p$. Interested readers may review \cite{liv,Peter} for other aspects of the deep theories of arithmetic geometry and number theory that were involved in the computation.
	
	Let
	\begin{eqnarray*} \label{5:x5}
		X:=\left\{(x_1,\ldots,x_5) \in \bP^4: \begin{array}{c}
			x_1+x_2+x_3+x_4+x_5=0\\
			 x_2x_3x_4x_5+x_1x_3x_4x_5+x_1x_2x_4x_5+x_1x_2x_3x_4=0\end{array}\right\}. \end{eqnarray*}
	$X$ is a surface. It was known that the ten points which form the $\mathfrak{S}_5$-orbit of $(1,-1,0,0,0)$ are all singular on $X$; in fact they are ordinary double points which become smooth after one blow up. The resulting surface $\tilde{X}$ is smooth, that is, we let
	\[\tilde{X}:=\mbox{the blow up of } X \mbox{ in the } \mathfrak{S}_5 \mbox{-orbit of } (1,-1,0,0,0). \]
	This surface is considered to be defined over $\mathbb{Q}$. It is non-singular and reduces to a non-singular surface $\tilde{X}(p)$ whenever $p \ne 3,5$.
	
	Let $q=p^m$. We denote by $N_q(Y)$ the number of points of a variety $Y$ with all coordinates in $\gf(q)$. Since every singular point is defined over the ground field $\gf(p)$ and is replaced by a projective line upon blowing up, we derive
	\begin{eqnarray} \label{5:xbx} N_q\left(\tilde{X}\right)=N_q\left(X\right)+10q.\end{eqnarray}
	It turns out that the K3-surface $\tilde{X}$ has maximal Picard number whose Hasse zeta function over $\Q$ can be explicitly computed, from which the quantity $N_q\left(\tilde{X}\right)$ can be derived easily. The result of $N_q\left(\tilde{X}\right)$ can be summarized below (see (0.1) Theorem, (4.12) Proposition and (4.12) Proposition (iii) in \cite{Peter}):
	\begin{theorem} \label{5:nqxb}
		Let $q=p^m$ where $p$ is a prime. Then
		\begin{eqnarray*} \label{5:aq}
			N_q\left(\tilde{X}\right)=1+q^2+q\left(16+4 \left(\frac{q}{3}\right)\right)+A_q,
		\end{eqnarray*}
		where $(\frac{\cdot}{\cdot})$ is the Jacobi symbol and
		\begin{itemize}
			\item[(1)] If $p=3$, then $A_q=(-1)^m \cdot q$;
			
			\item[(2)] If $p=5$, then $A_q=q$;
		\end{itemize}
		If $p \ne 3,5$ then
		\[A_q=\beta_1^m+\beta_2^m,\]
		where $\beta_1,\beta_2$ are reciprocal of the two complex roots of the polynomial
		\begin{eqnarray*} Q_2(t)=1-A_pt+\left(\frac{p}{15}\right)p^2t^2,\end{eqnarray*}
		and the quantity $A_p$ is the coefficient of $z^p$ in the $z$-expansion of
		\[\left(\sum_{m,n \in \Z} z^{m^2+mn+4n^2}\right) \cdot z \cdot \prod_{r=1}^{\infty} (1-z^r)(1-z^{3r})(1-z^{5r})(1-z^{15r}).\]
		Or in this case $A_p$ can be computed more precisely as follows:
		\begin{itemize}
			
			\item[(3)] If $p \equiv 7, 11,13,14 \pmod{15} ~ \Longleftrightarrow ~ \left(\frac{p}{15}\right)=-1$, then $A_p=0$, so $\beta_{1,2}=\pm p$ and
			\[A_q=q\left(1+(-1)^m\right).\]
			
			\item[(4)] If $p \equiv 1, 4 \pmod{15}$, let $a,b$ be integers satisfying $a^2+ab+4b^2=p$, then $A_p=2a^2-7b^2+2ab$.
			
			\item[(5)] If $p \equiv 2,8 \pmod{15}$, let $a,b$ be integers satisfying $2a^2+ab+2b^2=p$, then $A_p=a^2+8ab+b^2$.
		\end{itemize}
	\end{theorem}
	
Here we compute $A_q$ for several small $p$.

	\begin{example}
	When $p=7, 11,13$, then $A_q=q\left(1+(-1)^m\right)$. 
	\end{example}
	
	\begin{example}
		When $p=17$, since $p \equiv 2 \pmod{15}$ and since $17=2a^2+ab+2b^2$ for $a=-1,b=3$, we have $A_{17}=a^2+8ab+b^2=-14$, the reciprocal of the roots of $Q_2(t)=1+14t+17^2t^2$ are
		\[\beta_{1,2}=-7 \pm 4 \sqrt{-15}, \]
		so
		\[A_q=\left(-7 + 4 \sqrt{-15}\right)^m+\left(-7 - 4 \sqrt{-15}\right)^m. \]
	\end{example}
	
	\begin{example}
		When $p=19$, since $p \equiv 4 \pmod{15}$ and since $19=a^2+ab+4b^2$ for $a=1,b=2$, we have $A_{17}=2a^2-7b^2+2ab=-22$, the reciprocal of the roots of $Q_2(t)=1+22t+19^2t^2$ are
		\[\beta_{1,2}=-11 \pm 4 \sqrt{-15}, \]
		so
		\[A_q=\left(-11 + 4 \sqrt{-15}\right)^m+\left(-11 - 4 \sqrt{-15}\right)^m. \]
	\end{example}
	
	\begin{example}
		When $p=23$, since $p \equiv 8 \pmod{15}$ and since $23=2a^2+ab+2b^2$ for $a=1,b=3$, we have $A_{23}=a^2+8ab+b^2=34$, the reciprocal of the roots of $Q_2(t)=1-34t+23^2t^2$ are
		\[\beta_{1,2}=17 \pm 4 \sqrt{-15}, \]
		so
		\[A_q=\left(17 + 4 \sqrt{-15}\right)^m+\left(17 - 4 \sqrt{-15}\right)^m. \]
	\end{example}

	Armed with Theorem \ref{5:nqxb}, we can compute $N_{\mathbf{1^5}}$ for any prime $p$. First, let
	\begin{eqnarray} \label{5:hx} H:=\bigcup_j X \cap \left\{x_j=0\right\}, \quad X'=X \setminus H. \end{eqnarray}
	It is easy to check that
	\begin{eqnarray} \label{5:nqh} N_q(H)=10q-10,\end{eqnarray}
	accounting for the ten points which form the $\mathfrak{S}_5$-orbit of $(1,-1,0,0,0)$, and the $10(q-2)$ points which form the $\mathfrak{S}_5$-orbit of $(x,-x-1,1,0,0)$ for any $x \in \gf(q) \setminus \{0,-1\}$. 
	
	Next define
	\begin{eqnarray} \label{5:xtri} \triangle:=\bigcup_{i \ne j} X' \cap \{x_i=x_j\}, \quad X'':=X' \setminus \triangle. \end{eqnarray}
	Clearly by (\ref{5:xbx}) and (\ref{5:hx})--(\ref{5:xtri}) we have
	\begin{eqnarray} \label{ap:forN15} \frac{N_{\mathbf{1^5}}}{q-1}=N_q\left(X''\right)= N_q\left(\tilde{X}\right)-N_q(\triangle)-20q+10. \end{eqnarray}
	The quantity $N_q\left(\triangle\right)$ can be also treated easily: partitioning the projective variety $\triangle$ into disjoint subsets according to the ``pattern'' of indices $i,j$ such that $x_i=x_j$ (for example, the pattern $\mathbf{2^13^1}$ means that among the variables $x_1,x_2,\cdots, x_5$, exactly two variables are equal, the other three are equal, and these two are distinct, one example is $x_{1}=x_{2},x_{3}=x_{4}=x_{5}$ and $x_1 \ne x_3$),  we have
	\begin{eqnarray*} \label{5:tri}
		N_q\left(\triangle\right)&=&\frac{1}{q-1} \cdot \left\{N_{\mathbf{5^1}}+\binom{5}{1}N_{\mathbf{1^14^1}}+\binom{5}{2}N_{\mathbf{2^13^1}}+
		\binom{5}{3}N_{\mathbf{1^23^1}}
		+\binom{5}{2}N_{\mathbf{1^32^1}}+\binom{5}{1} \cdot 3 N_{\mathbf{1^12^2}}\right\}.
	\end{eqnarray*}
	Here the quantities such as $N_{\mathbf{5^1}}, N_{\mathbf{1^14^1}},N_{\mathbf{1^23^1}}$ etc are all given in Lemmas \ref{pre:N1}-\ref{pre:N6}. Now returning to (\ref{ap:forN15}) we obtain 
	\begin{eqnarray} \label{5:N1^5}
			\frac{N_{\mathbf{1^5}}}{q-1}= N_q\left(\tilde{X}\right)-20q+10-\delta_{p,5}-5\delta_{p,3}- \frac{10N_{\mathbf{1^23^1}}}{q-1}-\frac{10N_{\mathbf{1^32^1}}}{q-1}-
			\frac{15N_{\mathbf{1^12^2}}}{q-1},
	\end{eqnarray}

	Finally employing Theorem \ref{5:nqxb} and Lemmas \ref{pre:N1}-\ref{pre:N6}, we can summarize the result of $N_{\mathbf{1^5}}$ and $A_{\mathbf{1^5}}$ below.
	
	\begin{lemma} \label{5:A15} For any prime $p$, we have 
		\[\frac{N_{\mathbf{1^5}}}{q-1}=\left\{\begin{array}{ll}
		q^2+q\left(-14+(-1)^m\right)+36 & \mbox{ if } p=3\\
		q^2+q\left(-13+4(-1)^m\right)+10\left(7+(-1)^m\right)& \mbox{ if } p=5\\
		q^2+q \left(-14+4 \left(\frac{q}{3}\right)\right)+86+20\left(\frac{5}{q}\right)+15\left(\frac{-15}{q}\right)-10\lambda_q+A_q& \mbox{ if } p \ge 7\end{array}\right.\]
		\[ A_{\mathbf{1^5}}=\frac{1}{5!} N_{\mathbf{1^5}}.\]
		Here $\lambda_q$ is given by Theorem \ref{pre:thm1} and $A_q$ is given by Theorem \ref{5:nqxb}. 
	\end{lemma}

Now we can prove Lemma \ref{4:gamma25}.

\noindent {\bf Proof of Lemma \ref{4:gamma25}}
Since $\Gamma_2=A_{\mathbf{1^2}}+A_{\mathbf{2^1}}$, quoting Lemma \ref{pre:N1}, we obtain $\Gamma_2=\frac{q-1}{2}$. On the other hand, 
\[\Gamma_5:=A_{\mathbf{5^1}}+A_{\mathbf{1^14^1}}+ A_{\mathbf{2^13^1}}+ A_{\mathbf{1^23^1}}+ A_{\mathbf{1^32^1}}+A_{\mathbf{1^12^2}}+A_{\mathbf{1^5}}. \]
Quoting Lemmas \ref{pre:N1}-\ref{pre:N6} together with Lemma \ref{5:A15}, we can obtain the value of $\Gamma_5$ as described in (\ref{4:gamma5}). This completes the proof of Lemma \ref{4:gamma25}.  $\hfill \square$


\end{document}